\documentclass[11pt,twoside]{article}
\usepackage{inputenc}

\usepackage{fullpage}
\usepackage{ifthen}
\usepackage{epsf}
\usepackage{fancyheadings}
\usepackage{hyperref}
\usepackage{url}
\usepackage{enumerate}

 \usepackage{amsthm}
 \usepackage{amsfonts}
 \usepackage{amsmath}
 \usepackage{amssymb}
 \usepackage{color}
\usepackage{microtype}
\usepackage{multicol}
 \usepackage{amsfonts}
 \usepackage{amsmath}
 \usepackage{amssymb}

\usepackage{graphics}
\usepackage{graphicx}
\usepackage{mathrsfs}
\usepackage{float}

\usepackage{epsf}
 \usepackage{subfigure}
 \usepackage{caption}

\usepackage{amsmath}
\usepackage{verbatim}
\usepackage{fixltx2e}
\usepackage{amssymb}
\usepackage[ruled, vlined, linesnumbered]{algorithm2e}
\usepackage{color}
\usepackage{graphicx}
\usepackage{multirow}

\newtheorem{definition}{Definition}
\newtheorem{theorem}{Theorem}
\newtheorem{lemma}{Lemma}
\newcommand{\myparagraph}[1]{\noindent{\scshape \bfseries #1.}}
\newenvironment{protocol}[1][htb]
  {% Update algorithm name
   \begin{algorithm}[#1]%
  }{\end{algorithm}}  
\def\tbd{\textcolor{red}{[TBD]}}

\def\server{\mathcal{S}}
\def\client{\mathcal{C}}
\def\transaction{\mathcal{T}}
\def\itemset{\mathcal{I}}
\def\Algo1{\tbd Algorithm 1}

\newcommand{\e}[1]{\textcolor{blue}{#1}}

\newcounter{myctr}
\newenvironment{mylist}{\begin{list}{(\textbf{\arabic{myctr}})}
{\usecounter{myctr}
\setlength{\topsep}{1mm}\setlength{\itemsep}{0.5mm}
\setlength{\parsep}{1mm}
\setlength{\itemindent}{1mm}\setlength{\partopsep}{0mm}
\setlength{\labelwidth}{-2mm}
\setlength{\leftmargin}{0mm}}}{\end{list}}

\newlength{\widebarargwidth}
\newlength{\widebarargheight}
\newlength{\widebarargdepth}

\makeatletter
\long\def\@makecaption#1#2{
        \vskip 0.8ex
        \setbox\@tempboxa\hbox{\small {\bf #1:} #2}
        \parindent 1.5em 
        \dimen0=\hsize
        \advance\dimen0 by -3em
        \ifdim \wd\@tempboxa >\dimen0
                \hbox to \hsize{
                        \parindent 0em
                        \hfil 
                        \parbox{\dimen0}{\def\baselinestretch{0.96}\small
                                {\bf #1.} #2
                                } 
                        \hfil}
        \else \hbox to \hsize{\hfil \box\@tempboxa \hfil}
        \fi
        }
\makeatother

\title{Privacy-preserving Targeted Advertising}

\author{Theja Tulabandhula$^{1*}$, Shailesh Vaya$^{2}$ and Aritra Dhar$^{3}$\\
$^{1}$University of Illinois at Chicago\\
$^{2}$The Aquarian Inventors\\
$^{3}$ETH Zurich\\
$^{*}$email: theja@uic.edu}
\date{}

\begin{document}

\maketitle

\thispagestyle{plain}
\pagestyle{plain}

\begin{abstract}
Recommendation systems form the center piece of a rapidly growing trillion dollar online advertisement industry. Even with numerous optimizations and approximations, collaborative filtering (CF) based approaches require real-time computations involving very large vectors.
Curating and storing such related profile information vectors on web portals seriously breaches the user's privacy. Modifying such systems to achieve private recommendations further requires communication of long encrypted vectors, making the whole process inefficient. We present a more efficient recommendation system alternative, in which user profiles are maintained entirely on their device, and appropriate recommendations are fetched from web portals in an efficient privacy preserving manner. We base this approach on association rules.
\end{abstract}

%\keywords{Private recommendation; Privacy preserving protocols; Largest subset containment search; Association rules; Homo-morphic encryption}

\section{Introduction}
\label{sec:intro}

Targeted advertising (TA) uses keywords, their frequencies, link structure of the web, user interests/demographics, recent/overall buying histories etc. to deliver personalized  advertisements~\cite{korolova2010privacy}. TA is enabled by (unique) cookies (random hash-maps) stored on the user's devices~\cite{mayer2012third}. When a cookie is retrieved from it, the web-server storing this string can recall the profile of the user associated/stored with it at its end. This mechanism constitutes a serious breach of privacy as it allows the websites to build very elaborate profile of the user at their end~\cite{nyt,bit,steel2010facebook}. 
This leads to the question: can we achieve TA without employing the \emph{cookie} mechanism? Any such alternate approach would require the user to maintain their profile on their own device, and use it for interactively computing and fetching the appropriate TA from the server. In fact, solving this problem effectively can be considered a highly valuable contribution~\cite{udo2001privacy}. We present a solution framework for this very important commercial problem of wide-spread interest. In the presentation, we focus on the crux of the solution: a \textit{privacy preserving recommendation system}.

We present a novel approach for privacy preserving recommendation system (RS), which shifts most of the computational work to a \textit{pre-processing} stage. In the query processing phase, i.e., when we need to provide recommendations to a user according to their profile, most of the computation is done by the server in our proposed system, while the client merely computes and exchanges encryptions of a few messages. We will use recommender systems terminology in the paper, wherein items can also denote ads. Our system is based on selection and application of \textit{association rules} (AR), to produce an ordered list of recommended items. ARs \cite{agrawal1994fast} capture the relation that if a user has already bought a set of items $p$ (called the \emph{antecedent}), then she is very likely to also buy another set of items $q$ (called the \emph{consequent}). ARs are mined from large databases of historical user purchase data on the server-side, and often filtered to retain only the most meaningful insights~\cite{sarwar2000analysis}. While CF methods~\cite{sarwar2001item} extract and use pairwise marginal statistics from the joint probability distribution of user preferences over items, ARs take into account more complex summary statistics of the joint probability distribution, which capture information about much larger collections of items. Most computation in a CF based system is done both offline as well as online when the recommendation is computed. However, ARs are generated at an offline pre-processing stage, and query-processing merely requires selection and application of the right AR interactively. If done correctly, this process can have lower communication complexity in a privacy-preserving system. Lastly, note that CF encounters more bottlenecks when privacy preserving properties are desired: (a) Explicit and consistent feedback from users are hard to obtain, store, organize and compile, given that they are to be attached to the right user ID; and (b) the spatio-temporal profiles of users can change continuously leading to inaccurate targeted recommendations from a CF system. Because ARs are user-agnostic (they only track population level statistics), they can be more robust to these two issues.

In order to build a privacy preserving recommendation system using ARs, we first devise meaningful criteria for selection of the most relevant ARs given a query (user profile) and to form an ordered list of recommended items. We then present fast exact and approximate algorithms whose outputs satisfy these criteria. We provide their privacy preserving versions in which the client does not reveal its profile (we will use the word \emph{transaction} borrowed from the AR literature to denote a profile) to the server, and the server does not reveal its database of ARs to the client. Finally, we present experimental results to demonstrate the practicality of our solutions, in terms of the latency overheads incurred due to privacy, in e-commerce and other domains.

\subsection*{Problem Statement and Contributions}
There is a client $\client$, which has a set of items, referred to as a transaction set $\mathcal{T} \in\; 2^{\mathcal{I}}$, with $\mathcal{I}$ being the universal item set and $2^{\mathcal{I}}$ denoting the power set. The transaction set $\mathcal{T}$ represents the current profile of a user. A server $\mathcal{S}$ processes a large historical transaction database and stores a database of association rules $D$, in the form of $\{p_i \rightarrow q_i\}_{i=1}^{|D|}$. Here $|\cdot|$ represents the size of a set/database and $p_i$ and $q_i$ are sets (belonging to $2^{\mathcal{I}}$) that define the i$^{th}$ AR. In the \emph{non-private version} of the problem, given a single transaction $\mathcal{T}$, the server $\mathcal{S}$ computes and sends recommended items that are based on consequents of matching association rules, where \emph{matching} is defined suitably. For instance, if there are multiple ARs that are applicable to the transaction, then an ordered list of recommended items is prepared by collating the items recommended by each of the multiple ARs taking into consideration the (optional) weights attached to these rules (for instance, these can be lift, conviction, Piatetsky-Shapiro etc. or directly optimized for recommendation accuracy), or such ordering can also be based on the items in the input transaction, which may be assigned a weight according to some monotonously decreasing function of time lapsed since this item was active (e.g., purchased). In the \emph{privacy-preserving version}, given an input transaction $\transaction$, defined as an ordered list of items, held by the $\client$, and a database of ARs ($D$) held by the $\server$, the $\client$ and $\server$ privately and interactively compute the most relevant item, or ordered list of items, to be recommended to $\client$.

Below is a list of our contributions:
\begin{mylist}
\item \textbf{Criteria:} We formulate \textit{several criteria} for selecting the appropriate set of association rules. A rule is applicable if its antecedent is contained in the transaction. These criteria are differentiated using parameters such as threshold weight $w$, which is used to eliminate all rules below the given threshold weight, antecedent length threshold $t$, which is used to eliminate all rules above the given length threshold, and parameter $k$, which is used to select the \emph{top}-$k$ association rules under a specified ordering. We relate these criteria to a newly defined set optimization problem called the \emph{Generalized Subset Containment Search} (GSCS). We show how the GSCS problem is a strict super-set of the well known Maximum Inner Product Search (MIPS) criterion \cite{shrivastava2014asymmetric}.

\item \textbf{Algorithms:} We develop efficient exact and approximate algorithms for computing recommended items based on the criteria above. Exact implementations build on a novel two-level hashing based data structure that stores the ARs in a manner so that their antecedents can be appropriately matched, and corresponding consequent(s) can be efficiently fetched. The key benefit of the data structure is that it provides a weak form of privacy by itself, and is readily amenable to the privacy protocols mentioned below. Our implementations are parallelizable, and can exploit multi-threading machines. We also design a novel fast randomized approximation algorithm for GSCS that fetches applicable ARs based on Locality Sensitive Hashing (LSH) \cite{datar2004locality} and hashing based algorithms for MIPS \cite{neyshabur2015symmetric}.

\item \textbf{Privacy-preserving 2-party Protocols:}
  Next, we design communication efficient privacy preserving protocols corresponding to the above exact and approximate algorithms. These protocols are based on \emph{oblivious transfer} and straightforward to implement. Further, the protocol for the approximation algorithm can be easily extended to embed many other large scale data processing tasks that rely on LSH (for instance, record linkage, data cleaning and duplicity detection problems \cite{shrivastava2015asymmetric} to name a few). Finally, we  extensively evaluate the impact that adding privacy has in terms of latency in recommending items. We emphasize that these latencies are manageable for reasonably sized databases (e.g., $\sim 10^4$ ARs, see Section~\ref{sec:eval}) and practical for certain targeted advertising settings. We do note that achieving truly web-scale targeted advertising ($100-1000X$ larger problem instances) under the design choices made in our ARs based recommendation system, while not impossible, would need further research.
\end{mylist}

\subsection*{Related Work}
The GSCS problem introduced in Section \ref{sec:approxalgo} is similar to other popular search problems on sets including the Jaccard Similarity (JS) problem, and vector space problems such as the Nearest-Neighbor (NN) and the Maximum Inner Product Search (MIPS) problems.  For these related problems, solutions based on hashing techniques such as LSH are already available. For instance, one can find a set maximizing Jaccard Similarity with a query set using a technique called \textsc{Minhash}. The NN problem can be addressed using \textsc{L2LSH}~\cite{datar2004locality} and variants. The MIPS problem can be solved approximately using methods such as \textsc{L2-ALSH(SL)}~\cite{shrivastava2014asymmetric} and \textsc{Simple-LSH}~\cite{neyshabur2015symmetric} among others. Note that the GSCS problem is different from all of these, and also different from the similar sounding Maximum Containment Search problem defined in \cite{shrivastava2015asymmetric}. The problem in~\cite{shrivastava2015asymmetric} is equivalent to the MIPS problem while the former is not. In this paper, we give a new approximate algorithm to the GSCS problem using hashing techniques listed above.

Privacy preserving recommendation systems have been well studied in the past. For instance, privacy based solutions for different types of collaborative filtering systems have been proposed in \cite{Li:2011:PPC:2052138.2052428,McSherry:2009:DPR:1557019.1557090,Zhang:2006:PCF:1134707.1134742}. Roughly in that setting, the problem reduces to computing the dot product of a matrix with a vector of real numbers, where the (recommendation) matrix is possessed by the server and the client possesses the vector, and both the client and the server are interested in preserving the privacy of their data. Since embedding such schemes into a privacy protocol based on cryptography is difficult, many solutions resort to data modification and adding noise. For instance, in \cite{Agrawal:2001:DQP:375551.375602}, the authors propose a perturbation based method for preserving privacy in data mining problems. This approach is only applicable when one is interested in aggregate statistics and does not work when more fine-grained privacy is needed. In  \cite{Berkovsky:2007:EPP:1297231.1297234}, the authors propose a decentralized distributed storage scheme along with data perturbation to achieve certain notions of privacy in the collaborative filtering setting. For the same setting, a method based on perturbations is also proposed in \cite{Polat:2005:SCF:1066677.1066860} and \cite{Canny:2002:CFP:564376.564419}. In the paper \cite{aimeur2008alambic}, the authors proposed a theoretical approach for a system called Alambic which splits customer data between the merchant and a semi-trusted third party. The security assumption is that these parties do not collude to reveal customers' data. A major difference between our work and all these solutions is that we base our privacy solutions on cryptographic primitives (notably \emph{oblivious transfer}) and build specific protocols that work with association rules. This is attractive because ARs are already heavily used in practice for exploratory analysis in the industry. In particular, we propose one of the first practical distributed privacy preserving protocols for recommendation systems based on selection and application association rules. Note that privacy has also been well studied in the context of generation of association rules from historical transaction data~\cite{vaidya2002privacy,rizvi2002maintaining}, but not much for the problem of their selection and application in a recommendation or targeted advertising context.

%Yao's garbled circuit vs GMW's interactive computation per gate

Our work is closely related to the literature on secure two-party computation~\cite{hazay2010efficient,pinkas2009secure}. In this model, two players with independent inputs want to compute a function of the union of their inputs while not revealing their own inputs to the other party. GMW and Yao~\cite{goldreich1987play,yao1986generate} prove feasibility of secure two party computation (assuming honest-but-curious parties) and are based on a Boolean circuit computing the desired function, although because of their generality, they require a lot of communication. On the other hand, specialized protocols (for instance, our private protocols) for promiment classes of problems (for instance, targeted advertising or recommendations) are worth designing because they can reduce privacy overheads considerably. Similar to these protocols, our solutions also implement the same functionality as a trusted third party. Secure sorting is integral part of our solution, and has been previously studied in the literature~\cite{jonsson2011secure}. Note that our work is also distinct from Trusted Execution Environments (TEE) such as Intel SGX, ARM TrustZone etc. The latter provide the integrity and secrecy of computation by placing all executions on the isolated encrypted memory. Our work addresses not just the resultant secure two-party computation problem, but also allows dealing with large number of association rules via LSH based indexing, which sets its apart from the rest of the literature. Because LSH based retrieval is similar to a standard database query retrieval problem, we are able to design privacy protocols for approximate retrieval as well as exact retrieval of item recommendations using the same primitives.
These protocols consider only \emph{honest-but-curious} type of corruptions of the involved parties, in which the parties cannot deviate from the main protocol, but can try to glean whatever extra information they can using the transcripts of execution of the protocol.

\subsection*{Overview}	

In Sections~\ref{sec:criteria}, we present different criteria for the selection of applicable association rules, assuming that the rules have been mined beforehand from historical transaction data on the server-side. While the rules could be mined in a privacy preserving manner as well, we sidestep this aspect here (see related works above for prior solutions for this). While the accuracy (or any other information retrieval measure) of the recommendations will be dependent on the rule mining step in conjunction with the selection criteria, we take a decoupled approach in this paper to isolate the impact of introducing privacy on the recommendation system at query-time. So for the first step, we assume that an off-the shelf rule miner (such as SPMF that implements Apriori/FPGrowth) has already been used to generate rules, and for the second step, we proceed with designing rule selection criteria optimized for query-time performance. In this setting, rules can be optimized for recommendation performance in the first step itself by assuming a weighted decision-list model class and optimizing the weights of the model by minimizing a suitable recommendation error. A rule with a higher weight can signify higher importance and should be considered first while recommending, and our selection criteria take this into account while fetching the most applicable ARs (see \emph{ordering functions} in Section~\ref{sec:criteria}). We present new approximate and exact algorithms that implement the proposed criteria in Sections \ref{sec:approxalgo} and \ref{sec:exactalgo} respectively. We then develop their privacy preserving counterpart protocols in Section \ref{sec:private-protocols}. In Section~\ref{sec:eval} we describe experimental evaluations that validate the practicality of the proposed algorithms and their privacy preserving versions in moderate-scale recommendation systems. It is important to note that neither AR based nor CF based recommendations can dominate each other in terms of recommendation quality, and their performances depend on the specific datasets. Hence, we don't pursue this comparison in this paper. Our conclusions are presented in Section \ref{sec:conc}.

\section{Recommendation Criteria}
\label{sec:criteria}

We have a set of $D$ association rules $\{p_i \rightarrow q_i\}_{i=1}^{|D|}$ with $p_i,q_i \in 2^{\itemset}$, and additional attributes (for instance, a weight attribute $w_i$ could denote an interestingness measure or the rule's importance in recommendation quality).  Given a query transaction $\transaction \in 2^{\itemset}$ of the client, we perform two steps: a \textit{Fetch} operation followed by a \textit{Collate} operation. We propose multiple criteria for deciding which association rules should be fetched. We also define a search problem called the \textbf{Generalized Subset Containment Search (GSCS)} and relate it  to one of the criteria. Our algorithm in Section \ref{sec:approxalgo} solves the GSCS problem in a computationally efficient albeit approximate manner. In Section~\ref{sec:exactalgo}, we give linear-time query and space efficient exact algorithms that are easily adaptable to a private protocol design in Section~\ref{sec:private-protocols}.

\subsection*{Fetch Step: Selection of Rules}
We define that an association rule $i$ is \emph{applicable} to a transaction $\transaction$ if and only if $p_i \subseteq \transaction$. The selection criteria are differentiated based on parameters such as threshold weight $w \in \mathbb{Z}_{+}$, which is used to eliminate all applicable association rules below the threshold weight, $t \in \mathbb{Z}_{+}$, which is used to eliminate rules with antecedent lengths greater than $t$, and $k \in \mathbb{Z}_{+}$, which is used to select the top-$k$ association rules according to a predefined notion of ordering specified using an ordering function $f:\{1,...,D\} \rightarrow \mathbb{Z}$.\\

\myparagraph{Ordering Functions} The function $f$ determines which of the applicable rules are the top $k$ rules. We can define $f$ such that rules can be ordered according to: (a) their weights $w_i$ (e.g., $f(i) = w_i$), or (b) antecedent lengths $|p_i|$ (e.g., $f(i) = |p_i|$) or, (c) a combination of both (e.g., $f(i) = g_1(w_i) + g_1(w_{\textrm{max}})\cdot g_2(|p_i|)$, where $g_1$ and $g_2$ are strictly monotonic integer-valued functions and $w_{\text{max}} = \max_{i=1,..,|D|}w_i$). For a pair of rules with antecedents $p_1$ and $p_2$ and weights $w_1$ and $w_2$, this latter function has the following properties: (i)  If $|p_1| < |p_2|$, then $f(1) \leq f(2)$, and (ii) If $|p_1| = |p_2|$ and $w_1 \leq w_2$, then $f(1) \leq f(2)$. Another example of a combination ordering function is $f(i) = g_2(|p_i|) + g_2(|\itemset|)\cdot g_1(w_i)$, which prefers weights first and then lengths in case of ties.\\

\myparagraph{The Criteria} The \textsc{TOP}-\textsc{Assoc}($k,w,t,f$) criterion outputs a set of applicable rules base on three parameters and an ordering function $f$. Parameter $w$ filters out rules with weights $\leq w$. Parameter $t$ retains rules with antecedents of length $\leq t$. Parameter $k \in \mathbb{Z}_{+}$ controls the number of applicable rules that are finally output. We can write the following optimization formulation representing this criterion as follows: $\max_{x \in \{0,1\}^{|D|}} \;\;\sum_{i}^{|D|}x_i\cdot f(i)$ such that $\sum_{i}x_i \leq k,$ and $x_i \leq \min\{\mathbf{1}[|p_i| \leq t],\mathbf{1}[w_i \geq w],\mathbf{1}[p_i \subseteq \transaction]\}$. Here, $x_i$ is a binary decision variable that can take a value in $\{0,1\}$ and indicates whether an applicable rule is selected. The inequality $x_i \leq \min\{\mathbf{1}[|p_i| \leq t],\mathbf{1}[w_i \geq w],\mathbf{1}[p_i \subseteq \transaction]\}$ is a mathematical programming notation for the following constraint: set  variable $x_i$ to $1$ if and only if $|p_i| \leq t, w_i \geq w$ and $p_i \subseteq \transaction$. The term $\mathbf{1}[expr]$ denotes an indicator function that takes the value $1$ if the $expr$ is true, and $0$ otherwise. Since $x_i$ is constrained to be less than the minimum of three indicator functions, it can only take a value $1$ if all three evaluate to $1$. Otherwise, $x_i$ has to necessarily take the value $0$.

\textit{\textsc{TOP-1-Assoc}($f$) criterion} is a special case of \textsc{TOP}-\textsc{Assoc}($\allowbreak k,\allowbreak w,t,f$) where $k = 1$, $w = 0$ and $t = |\itemset|$. And \textsc{TOP-K-Assoc}($\allowbreak k,f$) \emph{criterion} is the special setting where $w = 0$ and $t = |\itemset|$. \textit{\textsc{ALL-Assoc}($w,t$) criterion} is another special case where $k = |D|$, in which case $f$ does not matter. Finally, under the \textsc{ANY}-\textsc{Assoc}($k,w,t$) criterion, the output contains at most $k$ applicable association rules with weights $\geq w$ and rule antecedent lengths $\leq t$. The corresponding optimization formulation is $\max_{x \in \{0,1\}^{|D|}} \;\;\sum_{i}x_i$ such that $\sum_{i}x_i \leq k,$ and $x_i\leq \min\{\mathbf{1}[|p_i| \leq t],\mathbf{1}[w_i \geq w],\mathbf{1}[p_i \subseteq \transaction]\}$. To summarize, \textsc{TOP}-\textsc{Assoc}, \textsc{TOP-1-Assoc}, \textsc{TOP-K-Assoc}, \textsc{ALL-Assoc}, and \textsc{ANY}-\textsc{Assoc} are some of the selection criteria we propose. \\

\myparagraph{Generalized Subset Containment Search} The \textsc{TOP-1-Assoc}($f$) criterion with an $f()$ that (all else being equal) prefers longer applicable rules leads to two new search problems: (a) the Largest Subset Containment Search (LSCS) problem, and (b) its generalization, the Generalized Subset Containment Search (GSCS) problem. 

\noindent\textit{Notation}: Until now, we used $p_i$, $q_i$ and $\transaction$ to denote item sets (and $|\cdot|$ denotes their size). With a slight abuse of notation, we will also denote their set characteristic vectors with the same symbols. A set characteristic vector has its coordinates equal to $1$ if the items are present in the set, and $0$ otherwise. Thus, in this case, $p_i$, $q_i$ and $\transaction$ are also binary valued vectors in the $|\mathcal{I}|$-dimensional space. Also, we denote the $j^{\text{th}}$ coordinate of $p_i$ as $p_i^j$, where $1 \leq j \leq |\mathcal{I}|$, and the $\ell_1$-norm of $p_i$ as $\|p_i\|_1$. The meaning of the symbols will be hopefully clear from the context, and will also be reiterated as needed.

\noindent\textit{Largest Subset Containment Search}: The problem $\max_{1 \leq i \leq |D|} \allowbreak \sum_{j=1}^{|\itemset|}\transaction^j \cdot p_i^j$ subject to $p_i^j \leq \transaction^j\;\;\textrm{ for } 1\leq j\leq |\itemset|$, attempts to find a set (i.e., a set characteristic vector) whose inner product with the set characteristic vector $\transaction$ is the highest among all sets that are subsets of $\transaction$. It is related to the \textsc{TOP-1-Assoc}($f$) criterion for certain ordering functions as shown below.

\begin{lemma} When $f(i) \propto |p_i|$ for all $1 \leq i \leq |D|$, the \textsc{TOP-1-Assoc}($f$) criterion is equivalent to the LSCS problem.
\label{lemma:norm-inner-prod-eqiv}
\end{lemma}

The above lemma holds because if $p_i \subseteq \transaction$, then $|p_i|$ is equal to $\sum_{j=1}^{|\itemset|}p_i^j\transaction^j$, which is the objective of LSCS. This connection to LSCS allows us to design a sub-linear time (i.e., $o(|D|)$) randomized approximate algorithm for fetching ARs (see Section~\ref{sec:approxalgo}) under the \textsc{TOP-1-Assoc} criterion (with appropriately chosen $f$) if we can come up with such algorithms for the LSCS problem. And the way we achieve this is by building on fast unconstrained inner-product search techniques~\cite{neyshabur2015symmetric} that solve the \text{Maximum Inner Product Search (MIPS)} problem, which is the following problem. Given a collection of ``database'' vectors $r_i \in \mathbb{R}^{d}, 1 \leq  i \leq |D|$ and a query vector $s \in \mathbb{R}^{d}$ ($d$ is the dimension), find a data vector maximizing the inner product with the query: $r^* \in \arg\max_{1 \leq i \leq |D|} \sum_{j=1}^{d}r_i^js^j$. 

MIPS and LSCS are not equivalent to each other. To see this, consider the following MIPS instance constructed to mimic an LSCS instance. Let $d = |\itemset|$. Let $r_i$ be equal to normalized antecedent vectors: $r_i = \frac{1}{\|p_i\|_1}p_i $ with $p_i \in \{0,1\}^{|\itemset|}$ for $ 1 \leq i \leq |D|$, and let query $s$ be equal to the set characteristic vector $\transaction \in \{0,1\}^{|\itemset|}$. The normalization in the definition of $r_i$s ensures that smaller length antecedents are preferred in the MIPS instance in order to mimic the subset containment or `applicable' property. Let $p_{\textrm{LSCS}}$ and $p_{\textrm{MIPS}}$ be the optimal solutions of the LSCS and MIPS instances constructed above. Then the following statements hold.

\begin{lemma} (1) [LSCS feasible implies MIPS optimal] If the LSCS instance is feasible (i.e., there exists at least one $p$ such that $p^j = 0 $ for all $j$ where $\transaction^j = 0$), then all feasible solutions of LSCS instance are optimal for the MIPS instance. (2) [MIPS optimal does not imply LSCS optimal] If there exist $p_1,p_2$ such that $\|p_1\|_1 < \|p_2\|_1$ such that $p_1^j = 0 $ and $p_2^j=0$ for all $j$ where $\transaction^j = 0$, then $p_1$ (as well as $p_2$) is optimal for the MIPS instance but is not optimal for the LSCS instance.
\label{ref:mips-lscs-relation1}
\end{lemma}

In other words, part (1) implies that if the LSCS instance is feasible, then there exists an optimal solution $p_{\textrm{MIPS}}$ for the constructed MIPS instance that has $p_{\textrm{MIPS}}^j = 0$ for all coordinates where $\transaction^j=0$. Part (2) implies that the optimal solutions of the MIPS instance are potentially feasible for the LSCS instance if they satisfy a condition. However, there is no guarantee that they will be optimal for the LSCS problem. In the worst case, there could be as many as O($2^{|\transaction|}$) optimal solutions for the MIPS instance but only a unique solution for the LSCS instance. 

\noindent\textit{Generalized Subset Containment Search}: LSCS can be generalized to get the GSCS problem, which is as follows: $\max_{1 \leq i \leq |D|} f'(i)\cdot \allowbreak \sum_{j=1}^{|\itemset|}\transaction^j \cdot p_i^j$ subject to $p_i^j \leq \transaction^j\;\;$ for $ 1\leq j\leq |\itemset|$, where $f'(i)$ is an ordering function. When $f'(i) \allowbreak = 1$, this is LSCS. In other words, for LSCS the ordering determining the top is not dependent on attributes such as the weight $w_i$ and is only dependent on the antecedent length (through the inner product term in the objective). On the other hand, GSCS can account for arbitrary ordering functions, especially when such functions capture task specific meaning (such as being related to recommendation accuracy for instance). Since the GSCS problem is more general, it is clear that the GSCS and the MIPS problems are also different.

\subsection*{Collate Step: Item Recommendations}
Once we have generated a set $\mathcal{L}$ of applicable association rules according to one of the criteria described above (assumed non-empty, otherwise we return a predefined list such as the list of globally most frequent items), we can compile a list of item recommendations in the following two ways:

\noindent\textit{Uncapacitated setting}: We simply return the union of consequent(s) $q_i$ of the association rules in $\mathcal{L}$; however, this list can be  potentially large.

\noindent\textit{Capacitated setting}: The client may have a constraint $k' << |\itemset|$ on the number of items it can recommend to the user. In this case, we derive associated weights $\tilde{w}_j$ for each item $j \in \cup_{i \in \mathcal{L}}\;q_i$ by adding up the weights $w_i$ of the rules where item $j$ is in the consequent $q_i$. The list of recommended items are then sorted according to these accumulated weights. If the bound $k' < |\cup_{i \in \mathcal{L}}q_i|$, we return the top $k'$ items from this sorted list, else we return all the items.

\section{Approximate Algorithms}
\label{sec:approxalgo}

We give approximate randomized algorithms that allow for sublinear time fetching of applicable association rules under two different criteria: \textsc{TOP-1-Assoc} and \textsc{TOP-K-Assoc}. We first introduce the basics of locality sensitive hashing that will be key to the design of our algorithms. Next, we propose our first algorithm for the \textsc{TOP-1-Assoc}($f$) criterion. And finally, we show how selecting rules for the \textsc{TOP-K-Assoc} can be reduced to finding solutions for multiple instance of the former allowing us to reuse the data structure construction and query processing steps.

\subsection*{Hashing based Data Structures for Rule Retrieval}

LSH~\cite{andoni2008near,har2012approximate} is a technique for finding vectors from a known set of vectors that are  most `similar'  to a given query vector in an efficient manner (using randomization and approximation). We will use it to solve MIPS and GSCS problems associated with the criteria mentioned above. The method uses hash functions that have  the following locality-sensitive property.

\begin{definition}\label{def:hashfunctions}
A $(r,cr,P_1,P_2)$-sensitive family of hash functions ($h \in \mathcal{H}$) for a metric space $(X,d)$ satisfies the following properties for any two points $p,q \in X$:
\begin{itemize}
\item If $d(p,q) \leq r$, then $Pr_{\mathcal{H}}[h(q) = h(p)] \geq P_1$, and
\item If $d(p,q) \geq cr$, then $Pr_{\mathcal{H}}[h(q) = h(p)] \leq P_2$.
\end{itemize}
\end{definition}

LSH solves the nearest neighbor problem via solving a near-neighbor problem defined below.

\begin{definition} The $(c,r)$-NN (approximate \emph{near} neighbor) problem with failure probability $f \in (0,1)$ is to construct a data structure over a set of points $P$ that supports the following query:  given point $q$, if $\min_{p \in P}d(q,p) \leq r$, then report some point $p' \in P\cap \{p: d(p,q) \leq cr\}$ with probability $1-f$. Here, $d(q,p)$ represents the distance between points $q$ and $p$ according to a metric that captures the notion of neighbors. Similarly, the $c$-NN (approximate \emph{nearest} neighbor) problem with failure probability $f \in (0,1)$ is to construct a data structure over a set of points $P$ that supports the following query: given point $q$, report a $c$-approximate nearest neighbor of $q$ in $P$ (i.e., return $p'$ such that $d(p',q) \leq c\min_{p \in P}d(p,q)$) with probability $1-f$.
\end{definition}

The following theorem states that we can construct a data structure that solves the approximate near neighbor problem in sub-linear time.

\begin{theorem} [\cite{har2012approximate} Theorem 3.4] Given a $(r,cr,P_1,P_2)$-sensitive family of hash functions, there exists a data structure for the $(c,r)$-NN (approximate near neighbor problem) over points in the set $P$ (with $|P| = N$) such that the time complexity of returning a result is $O(nN^\rho/P_1 \log_{1/P_2}N)$ and the space complexity is $O(nN^{1+\rho}/P_1)$. Here $\rho = \frac{\log 1/P_1}{\log 1/P_2}$. Further, the failure probability is upper bounded by $1/3 + 1/e$.\label{thm:near-neighbor}
\end{theorem}

The failure probability above can be changed to any desired value through \emph{amplification}. The data structure is as follows~\cite{andoni2008near}: we employ multiple hash functions to increase the confidence in reporting near neighbors by amplifying the gap between $P_1$ and $P_2$. The number of such hash functions is determined by parameters $L_1$ and $L_2$. We choose $L_2$ functions of dimension $L_1$, denoted as $g_j(q) = ( h_{1,j}(q) , h_{2,j}(q), \cdots h_{L_1, j}(q) )$, where $h_{t,j}$ with $1 \leq t \leq L_1 ,1 \leq j \leq L_2$ are chosen independently and uniformly at random from the family of hash functions. The data structure for searching points with high similarity is constructed by taking each point $x$ (in our setting, these would be the set characteristic vectors of antecedents) and storing it in the location (bucket) indexed by  $g_j(x) , 1 \leq j \leq L_2$.  When a new query point $q$ is received,  $g_j(q) , 1 \leq j \leq L_2$ are  calculated and all the points from the search space in the buckets  $g_j(q) , 1 \leq j \leq L_2$ are retrieved.  We then compute the similarity of these points with  the query vector in a sequential manner and return any point that has a similarity greater than the specified threshold $r$. We also interrupt the search after finding the first $L_3$ points including duplicates (this is necessary for the guarantees in Theorem~\ref{thm:near-neighbor} to hold). Choosing $L_1 = \log N$, $L_2 = N^\rho$ and $L_3 = 3L_2$ allows for sublinear query time. The storage space for the data structure is $O(nN^{1+\rho})$, which is not too expensive (we need $O(nN)$ space just to store the points).

The following theorem states that an approximate nearest neighbor data structure can be constructed using an approximate near neighbor data structure.

\begin{theorem} [\cite{har2012approximate} Theorem 2.9] Let P be a given set of N points in a metric space, and let $c, f \in (0,1)$ and $\gamma \in (\frac{1}{N},1)$ be parameters. Assume we have a data-structure for the $(c,r)$-NN (approximate \emph{near} neighbor) problem that uses space $S$ and has query time $Q$ and failure probability $f$. Then there exists a data structure for answering $c(1+O(\gamma))$-NN (approximate \emph{nearest} neighbor) problem queries in time $O(Q\log N)$ with failure probability $O(f\log N)$. The resulting data structure uses $O(S/\gamma \log^2 N)$ space.\label{thm:nearest-neighbor}
\end{theorem}

Instead of using the data structure above, below we will use a slightly sub-optimal data structure (see \textsc{Approx-GSCS-Prep}) but amenable to private protocols (Section~\ref{sec:private-protocols}) as follows: We create multiple near neighbor data structures as described in Theorem~\ref{thm:near-neighbor} using different threshold values ($r$) but with the same success probability $1-f$ (by amplification for instance). When a query vector is received, we calculate the near-neighbors using the hash structure with the lowest threshold. We continue checking with increasing value of thresholds till we find at least one near neighbor. Let $\widetilde{r}$ be the first threshold for which there is at least one near neighbor. This implies that the probability that we don't find the true nearest neighbor is at most $f$ because the near neighbor data structure with the threshold $\widetilde{r}$ has success probability $1-f$. If the different radii in the data structures are not too far apart so that not too many points are retrieved, we can still get sublinear query times.

\subsection*{Solving for \textsc{TOP-1-Assoc}($f$)}

For the \textsc{TOP-1-Assoc} criterion, our goal is to return a single rule whose antecedent set is contained within the query set $\transaction$ (applicable) and whose $f(i)$ is the highest among such applicable rules. Our algorithm solves the GSCS formulation of this criterion by constructing a corresponding approximating MIPS instance, and solving it using locality sensitive hashing (LSH) based techniques\cite{neyshabur2015symmetric}. 

Our scheme has two parts: (a) Preprocessing state involving \textsc{Approx-GSCS-Prep} (Algorithm \ref{algo:Approx-GSCS-Prep}), and (b) Query stage involving \textsc{Approx-GSCS-Query} (Algorithm \ref{algo:Approx-GSCS-Query}). In \textsc{Approx-GSCS-Prep}, the algorithm prepares a data structure based on all rules that can be efficiently searched at query time to obtain applicable association rules. In \textsc{Approx-GSC-Query}, the algorithm takes a transaction $\transaction$ and outputs the rule that satisfies the \textsc{TOP-1-Assoc}($f$) criteria via a linear-scan (worse case $O(2^{|\transaction|})$).\\  %decoupling applicability with the criteria. no need for functions such as C + alpha Log |p| - |P|

\myparagraph{Pre-processing Stage} The GSCS instance for obtaining applicable rules is: $\max_{1 \leq i \leq |D|} f(i)\cdot\sum_{j=1}^{|\itemset|}p_i^j\transaction^j$ such that $p_i^j \leq \transaction^j$ for all $1 \leq i \leq |D|, 1 \leq j \leq |\itemset|$. Since a GSCS instance cannot be transformed exactly to an MIPS instance as detailed in Section~\ref{sec:criteria}, we create an approximating MIPS that results in a set of candidate rules that contain the rule which is optimal for the \textsc{TOP-1-Assoc} criterion. From these rules, we prune for the one that satisfies the \textsc{TOP-1-Assoc}($f$) criterion by evaluating the ordering function at query time (see below). 

The approximate MIPS instance we construct is as follows: $\max_{1 \leq i \leq |D|} \frac{1}{\|p_i\|_{1}}\cdot\sum_{j=1}^{|\itemset|}p_i^j\transaction^j$, where we have replaced the hard constraints related to subset containment with a proxy scaling coefficient. The effect of the coefficient is that it prefers applicable antecedents with small antecedent length, which ensures that the chosen rule obeys the original containment constraint (see Lemma~\ref{ref:mips-lscs-relation1}). The objective of the MIPS instance can be viewed as an inner product between two real vectors, where the first vector is $\frac{1}{\|p_i\|_{1}}p_i$. Let vector $p_i' = \frac{1}{\|p_i\|_1}p_i \in \mathbb{R}^{|\itemset|}$, which satisfies $\|p_i'\|_2 \leq 1$. This re-parameterization achieves two things: (a) using $p_i'$ in the MIPS instance has the same effect as using $\frac{1}{\|p_i\|_{1}}p_i$, and (b) its $\ell_2$ norm is smaller than $1$, allowing us to apply the technique proposed in \cite{neyshabur2015symmetric} to build our data structure using
\textsc{Approx-GSCS-Prep}. This algorithm (shown in Algorithm \ref{algo:Approx-GSCS-Prep}) uses \textsc{Simple-LSH-Prep} (see Algorithm~\ref{algo:Simple-LSH-Prep}, proposed in~\cite{neyshabur2015symmetric}) to create the desired data structure that allows for sublinear time querying of applicable association rules. 

Our scheme also works on the LSH principle described above. The subroutine that we use, namely \textsc{Simple-LSH-Prep} (Algorithm~\ref{algo:Simple-LSH-Prep}), relies on the inner product similarity measure and comes with corresponding guarantees on the retrieval quality. To construct a data structure for fast retrieval, it uses hash functions parameterized by spherical Gaussian vectors $a \sim \mathcal{N}(0,I)$ such that $h_a(x) = \textrm{sign}(a^Tx)$ ($\textrm{sign}()$ is a scalar function that outputs $+1$ if its argument is positive and $0$ otherwise). 

Given the scaled $p_i'$ vectors, \textsc{Simple-LSH-Prep} constructs a data structure $\mathcal{DS}$ as follows. It defines a mapping $P$ for vector $x \in \{x \in \mathbb{R}^{|\itemset|}: \|x\|_2 \leq 1 \}$ as: $P(x) = \left[x;\sqrt{1 - \|x\|_2^2}\right] \allowbreak \in \mathbb{R}^{|\mathcal{I}|+1}$. Then, for any $p_i'$, due to our scaling, we have $\|P(p_i')\|_2 \allowbreak = 1$. Let $\transaction' = \frac{1}{\|\transaction\|_2}\transaction$ be the scaled version of the transaction $\transaction$. Then, for any scaled vector $p_i'$ and $\transaction'$ we have the following property: 
$P(p_i')^TP(\transaction') 
= \frac{1}{\|p_i\|_1}\sum_{j=1}^{|\itemset|}p_i^j\transaction'^j,$ 
where $()^{T}$ represents the transpose operation. This implies that the inner product in the space defined by the mapping $P$ is equal to our MIPS instance objective. Further, in this new space, maximizing inner product is the same as minimizing Euclidean distance. Thus, using the hash functions $\{h_a\}$ defined earlier to perform fast Euclidean nearest neighbor search achieves doing an inner product search in the original space defined by the domain of $P()$. 

Let the approximation guarantee for the nearest neighbor obtained be $1+\nu$ (e.g., set the parameters of the data structure in Theorem~\ref{thm:nearest-neighbor} such that it solves the $(1+ \nu)$-NN problem). Then the following straightforward relation gives the approximation guarantee for the MIPS problem.
\begin{lemma}\label{lemma:nn2mips}
If we have an $1+\nu$ solution $x$ to the nearest neighbor problem for vector $y$, then $1 + (1+\nu)^2(\max_{p \in P}p\cdot y - 1) \leq x \cdot y \leq \max_{p \in P}p\cdot y$.
\end{lemma}

\textsc{Simple-LSH-Prep} uses multiple parameters, including concatenation parameters $\{K_m\}_{m=1}^{M}$, and repetition parameters $\{L_m\}_{m=1}^{M}$ (these are determined by a sequence of increasing radii $\{r_m\}_{m=1}^{M}$ needed for the nearest-neighbor problem, see above). For every $m$, \textsc{Simple-LSH-Prep} picks a sequence of $K_m\cdot L_m$ hash functions from $\{h_a\}$ and gets a $K_m\cdot L_m$ dimensional signature for each vector $P(p_i') \in \mathbb{R}^{|\itemset|+1},\;\; 1 \leq i \leq |D|$. These signatures and the chosen hash functions (across all radii) are output as $\mathcal{DS}$.

%%%%%%%%%%%%%%%%%%%%%%%%%%%%%%%%%%%%%%%%%%%%%%%%%%%%%%%%%%%%%%%%%%%%%%%%%%%%%%%%%%%%
\begin{algorithm}[h]
\scriptsize
\DontPrintSemicolon
\KwIn{$D$ rules with antecedents $p_i \in \{0,1\}^{|\itemset|}$.}
\KwOut{Data Structure $\mathcal{DS}$ containing rule representations and hashing constants}
\Begin
{
	\textbf{Do parallel}\{\\
	\ForAll{$i=1,...,|D|$}
	{
		$p_i' \leftarrow \frac{1}{\|p_i\|_1}p_i$\\
	}
	\}\textbf{End parallel}\\
	$\mathcal{DS} \leftarrow$ \textsc{Simple-LSH-Prep}($\{p_i'\}_{i=1}^{|D|}$)\\
	}
	return $\mathcal{DS}$
\caption{\bf \textsc{Approx-GSCS-Prep}({\small $\{p_i\}_{i=1}^{|D|}$})}
\label{algo:Approx-GSCS-Prep}
\end{algorithm}
\setlength{\textfloatsep}{0pt}
%%%%%%%%%%%%%%%%%%%%%%%%%%%%%%%%%%%%%%%%%%%%%%%%%%%%%%%%%%%%%%%%%%%%%%%%%%%%%%%%%%%%

%%%%%%%%%%%%%%%%%%%%%%%%%%%%%%%%%%%%%%%%%%%%%%%%%%%%%%%%%%%%%%%%%%%%%%%%%%%%%%%%%%%%
\begin{algorithm}[h]
\scriptsize
\DontPrintSemicolon
\KwIn{Vectors $v_i \in \mathbb{R}^{d}$ for $1 \leq i \leq |D|$.}
\KwOut{Data structure $\mathcal{DS}$ capturing the hash functions and the hash signatures of input ARs}
\Begin
{
	\textbf{Preprocessing: Generate Hash functions}\\
	$\mathcal{G}  \leftarrow \phi$\\
	\textbf{Do parallel}\{\\
	\ForAll{$m=1,...,M$}
	{
		\ForAll{$l=1,...,L_m$}
		{
			\ForAll{$k=1,...,K_m$}
			{
				$a \in \mathbb{R}^{d+1} \sim \mathcal{N}(0,I)$\\
				$\mathcal{G}[m,l,k] \leftarrow a$
			}
		}
	}
	\}\textbf{End parallel}\\
	\textbf{Preprocessing: Hash data vectors}\\
	Define $P(x) = \left[x;\sqrt{1 -\|x\|_2^2}\right]$ for $x \in \{x \in \mathbb{R}^d: \|x\|_2 \leq 1\}$\\
	$\mathcal{H} \leftarrow \phi$\\
	\textbf{Do parallel}\{\\
	\ForAll{$1 \leq i \leq |D|$}
	{
		\ForAll{$m=1,...,M$}
		{
			\ForAll{$1 \leq l \leq L_m$}
			{
				$\textrm{index} \leftarrow \phi$ \\
				\ForAll{$1 \leq k \leq K_m$}
				{
					\eIf{$ \mathcal{G}[m,l,k]^TP(v_i) \geq 0 $}{
						$\textrm{index.append}(1)$
						}{
						$\textrm{index.append}(0)$
						}
				}
				$\mathcal{H}[m,l,\textrm{index}].add(i)$
			}
		}
	}
	\}\textbf{End parallel}\\

	$\mathcal{DS} \leftarrow (\mathcal{G},\mathcal{H})$\\

	return $\mathcal{DS}$

}
\caption{\bf \textsc{Simple}-\textsc{LSH}-\textsc{Prep}($\{v_i\}_{i=1}^{|D|}$)}
\label{algo:Simple-LSH-Prep}
\end{algorithm}
\setlength{\textfloatsep}{0pt}
%%%%%%%%%%%%%%%%%%%%%%%%%%%%%%%%%%%%%%%%%%%%%%%%%%%%%%%%%%%%%%%%%%%%%%%%%%%%%%%%%%%%

\myparagraph{Query Stage} Given a transaction $\transaction$, \textsc{Approx-GSCS-Query} queries \textsc{Simple-LSH-Query} (see Algorithm~\ref{algo:Simple-LSH-Query}) to obtain rules that are applicable. \textsc{Simple-LSH-Query} solves the approximating MIPS problem in sub-linear time by filtering out the most similar neighbors of transaction vector $\transaction$ from the set $\{p_i': 1 \leq i \leq |D|\}$. It does this by constructing the vector $[\transaction; 0] \in \mathbb{R}^{|\itemset|+1}$ and getting (for each $m=1,...,M$ sequentially) its $K_m\cdot L_m$ dimensional signature with the same hash functions as used before for $p_i'$ vectors. Then it collects rules that share the same signature as the transformed transaction vector (it stops this process at the first $m$ for which it is able to find a rule). In particular, it ensures that the signatures agree in at least one $K_m$-length chunk out of the $L_m$ chunks. Appropriate choices of $K_m$ and $L_m$ (which do not depend on the number of rules $|D|$) allows for retrieval of the top candidates with high approximation quality. \textsc{Approx-GSCS-Query} processes this candidate list to get the top rule in terms of the ordering function $f$. This is a linear search with worst case time complexity O($2^{|\transaction|}$). In case the candidate list is empty, it returns a predefined baseline rule.

Note that scaling $\transaction$ to $\transaction' = \frac{1}{\|\transaction\|_2}\transaction$ before passing it to transformation map $P()$ is not necessary at query time. This is because, Given a Gaussian vector $a \in \mathbb{R}^{d+1}$ and a transaction vector $\transaction \in \{0,1\}^{d}$,
$\textrm{sign}(a^TP(\frac{\transaction}{\|\transaction\|_2}))\allowbreak = \textrm{sign}(a^T[\transaction;0])$. The advantage of this change is that we do not have to work with a real-valued vector at query time, leading to an efficient oblivious transfer (OT) step in the privacy preserving counterpart protocol that embeds this method (see Section \ref{sec:private-protocols} for more details).

%%%%%%%%%%%%%%%%%%%%%%%%%%%%%%%%%%%%%%%%%%%%%%%%%%%%%%%%%%%%%%%%%%%%%%%%%%%%%%%%%%%%
\begin{algorithm}[h]
\scriptsize
\DontPrintSemicolon
\KwIn{Data structure $\mathcal{DS}$ from \textsc{Approx-GSCS-Prep}, query $\transaction \in \{0,1\}^{|\itemset|}$}
\KwOut{Top most rule according to ordering function $f$}
\Begin
{
	Set $\mathcal{S} \leftarrow \phi$ \\
	$\mathcal{S} \leftarrow$ \textsc{Simple-LSH-Query}($\transaction,\mathcal{DS}$)\\

	\If{$\mathcal{S} = \phi$}{
		return a pre-defined default rule
	}

	Set $\mathcal{S'} \leftarrow \phi$\\
	Set $fval \leftarrow 0$\\ 
	\ForAll{$i \in \{i: p_i' \in \mathcal{S}\}$}{
		\If{$f(i) \geq fval$}{
			$\mathcal{S'} \leftarrow \{p_i \rightarrow q_i\}$\\
			$fval \leftarrow f(i)$\\
		}
	}
	return $\mathcal{S'}$\\
}
\caption{\bf \textsc{Approx-GSCS-Query}($\transaction,\mathcal{DS},f$)}
\label{algo:Approx-GSCS-Query}
\end{algorithm}
\setlength{\textfloatsep}{0pt}
%%%%%%%%%%%%%%%%%%%%%%%%%%%%%%%%%%%%%%%%%%%%%%%%%%%%%%%%%%%%%%%%%%%%%%%%%%%%%%%%%%%%

%%%%%%%%%%%%%%%%%%%%%%%%%%%%%%%%%%%%%%%%%%%%%%%%%%%%%%%%%%%%%%%%%%%%%%%%%%%%%%%%%%%%
\begin{algorithm}[h]
\scriptsize
\DontPrintSemicolon
\KwIn{Query $u \in \mathbb{R}^{d}$, data structure $\mathcal{DS}$}
\KwOut{Vector(s) with high inner products with query $u$}
\Begin
{
	$\mathcal{G},\mathcal{H} \leftarrow \mathcal{DS}$\\
	\textbf{Query:}\\
	\ForAll{$m=1,...,M$}
	{
		Set $\mathcal{S} \leftarrow \phi$\\
		\ForAll{$1 \leq l \leq L_m$}
		{
			index$ \leftarrow \phi$\\
			\ForAll{$1 \leq k \leq K_m$}
			{
				\eIf{$ \mathcal{G}[m,l,k]^T[u;0] \geq 0 $}{
					index.append(1)
					}{
					index.append(0)
					}
			}

			$\mathcal{S}$.add($\mathcal{H}$[m,l,index])\\

			\If{$|\mathcal{S}| > 3L_m$}{
				break\\
			}
		}

		\If{$\mathcal{S} \neq \phi$}{
			% $dval \leftarrow MAX$\\
			% $ibest \leftarrow nil$\\
			% \ForAll{$1 \leq i \leq |\mathcal{S}|$}
			% {
			% 	\If{$\|p'_i - u\|_2 \leq dval$}{
			% 		$dval \leftarrow \|p'_i - u\|_2$\\
			% 		$ibest \leftarrow i$\\
			% 	}
			% }
			% return $ibest$
			return $\mathcal{S}$
		}
	}

	return $\mathcal{S}$
}
\caption{\bf \textsc{Simple}-\textsc{LSH}-\textsc{Query}($u$,$\mathcal{DS}$)}
\label{algo:Simple-LSH-Query}
\end{algorithm}
%%%%%%%%%%%%%%%%%%%%%%%%%%%%%%%%%%%%%%%%%%%%%%%%%%%%%%%%%%%%%%%%%%%%%%%%%%%%%%%%%%%%

\subsection*{Solving for \textsc{TOP-K-Assoc}($k,f$)}

The approximate algorithm above can be adapted to the \textsc{TOP-K-Assoc} criterion due to a reduction from the approximate $k$-nearest neighbor problem and the approximate $1$-nearest neighbor problem (the reduction and its analysis are due to Sariel Har-Peled, 2018). The reduction is as follows. Given database $D$ and the parameter $k$, we construct $N = k\log|D|$ copies of the database ($D_1,...,D_N$) where in each database, every rule is included with a constant probability $1/k$. Given these $N$ databases, we apply \textsc{Approx-GSCS-Prep} to each to generate $\mathcal{DS}_1,...,\mathcal{DS}_N$. When we want to run a query $\transaction$ to get the top-$k$ applicable association rules according to an ordering function $f$, we seek the most highly applicable rule from each of the data structures using \textsc{Approx-GSCS-Query}. Once we retrieve $N$ highly applicable rules, we then prune this list by linear scanning and sorting to obtain the top-$k$ rules. This reduction increases the query time roughly by a factor that is linear in $k$ and logarithmic in the number of rules $|D|$. 

An intuitive argument for the reduction (Anastasios Sidiropoulos, 2018) is the following: Let $X$ be the set of $k$-nearest neighbors (rules satisfying the given criterion) to the query $\transaction$. When sampling a subset of the rules, for any $x \in X$, with probability $\Theta(1/k)$, we include $x$ and exclude every other rule in $X$. The specified number of sampled copies of the database $D$ are just enough to recover the top-$k$ rules with high probability, even when there are approximations.
\section{Exact Algorithms}
\label{sec:exactalgo}

For exact retrieval of applicable association rules according to any of the criteria in Section \ref{sec:criteria}, we essentially perform a linear scan over all rules, filter them according to the appropriate thresholds and sort them according to the given ordering function. Our main contribution here is a two-level data structure to store the ARs that has two attractive properties: (a) It can \emph{efficiently} store the rules for fetching quickly (in terms of the overall communication complexity and computational complexity needed), and (b) the data structure is easy to \emph{privatize} for use in a privacy-preserving protocol (see Section \ref{sec:private-protocols}). The data structure (denoted as $\mathcal{H}$) is common to exact implementations of all the criteria specified in Section~\ref{sec:criteria}. We describe the data structure in a generic way for retrieval of strings. Adapting the notion of strings to rules (and their attributes) in our setting is straightforward.

\subsection*{Pre-processing Stage}  

Consider a database $D$ of strings, with each string of maximum length $M$ (i.e., each string is a sequence of symbols from some ground set $\Sigma$).  We generate a data structure $\mathcal{H}$ that stores these strings using \textsc{Exact-Fetch-Prep} (Algorithm~\ref{data-structure}). The structure is an adaptation of~\cite{FKS}, but unlike \cite{FKS} it is symmetric and has two levels. By \emph{symmetric} we mean that a fixed hash function ($h_r$) will be chosen for hashing all elements at the first level, and another hash function ($h_s$) is chosen for hashing all elements at the second level. Such a choice helps with the complexity of oblivious transfer (OT) protocol in Section ~\ref{sec:private-protocols}, where we will essentially use the same data structure. That is, using the same data structure, the client can compute encrypted indices on its end with the knowledge of ($h_r,h_s$) and can use OT to retrieve objects, and efficient implementations for this already exist in practice. The hash functions map elements from $[|D|]$ (for a number $N$ the notation $[N]$ represents the set ${1,...,N}$) to a range of size $L = 16 \cdot |D|$, and are chosen randomly from a $2$-Universal hash function family~\cite{Carter:1977:UCH:800105.803400} $\mathcal{H}_2 = \{h: [U] \rightarrow [L]\}$, where $U$ is a large positive integer (note that these are not locality-sensitive). 

The choice of a two-level hashing is inspired by the treatment in~\cite{FKS} where the authors show that a two level hashing can simultaneously lead to linear storage complexity as well as constant worst case retrieval time complexity compared to single level hashing (where one typically trades off storage space vs retrieval time). In addition, to avoid collisions, the range of a single hash would have to be very large. On the other hand, the two level structure does not require the first hash function to be collision-free, and this helps with the storage vs retrieval tradeoff.

Additionally, we first choose two large integers $r$ and $l$, a string $r'$ of length $l$ and the MD5 hash function \cite{MD5-reference} (denoted as $C_r:\Sigma^{M+l} \rightarrow [2^r]$) to transform the database strings. Once $\mathcal{H}$ is created on the server, it publicly declares the hash functions $h_r$ and $h_s$ as well as the constant string $r'$ it generated. Details of the construction of $\mathcal{H}$ using \textsc{Exact-Fetch-Prep} are shown in Algorithm~\ref{data-structure}.

%%%%%%%%%%%%%%%%%%%%%%%%%%%%%%%%%%%%%%%%%%%%%%%%%%%%%%%%%%%%%%%%%%%%%%%%%%%%%%%%%%%%
\SetKwRepeat{Do}{do}{while}%
\begin{algorithm}[h]
\scriptsize
\DontPrintSemicolon
\KwIn{Database $D$}
\KwOut{Hash table with two level hashes $\mathcal{H}$, hashing functions $h_r,h_s$ and string $r'$}
\Begin
{
  Choose: (a) large positive integers $r$ and $l$, (b) arbitrary string $r'$ of length $l$, and (c) collision resistant cryptographic hash function $C_r: \Sigma^{l+M} \rightarrow 2^r$.\\
  $D_e \leftarrow \phi$\\
  \ForAll{$x' \in D$}
  {
    $x \leftarrow r'\circ x'$ ($\circ$ denotes the concatenation operator)\\
    $D_e \leftarrow D_e \cup \{x\}$\\
  }

  \Do{ $\sum^L_{i=1}{b^2_i} \leq 4 |D|$}{ 

    $h_r \sim$ Uniform($\mathcal{H}_2$) \\
    \ForAll{$i=1,...,L$}
    {
      $B_i = \{x \in D_e :h_r(x)=i\}$\\
      $b(i) \leftarrow |B_i|$
    }
  }

  \Do{$\forall\; 1 \leq i \leq L$ and $\forall\; x,y \in B_i, h_s(x) \neq h_s(y)$ }{ 

    $h_s \sim$ Uniform($\mathcal{H}_2$)
  }

  Initialize array $\mathcal{H}$ of size $L^2 + L$.\\
  \ForAll{$x \in D_e$}{
    $\mathcal{H}$[$L \cdot h_r(x) + h_s(x)$] = $\{C_r(x) \textrm{ and other data associated with } x\}$  
 
  }

  return $(\mathcal{H},h_s,h_r,r')$
}

\caption{\textbf{\textsc{Exact-Fetch-Prep}($D$)}: Creating two-level data structure $\mathcal{H}$}
\label{data-structure}
\end{algorithm}
\setlength{\textfloatsep}{0pt}
%%%%%%%%%%%%%%%%%%%%%%%%%%%%%%%%%%%%%%%%%%%%%%%%%%%%%%%%%%%%%%%%%%%%%%%%%%%%%%%%%%%%

We now show two properties that \textsc{Exact-Fetch-Prep} satisfies. First, it identifies the two hash functions $h_r, h_s$ that lead to no collisions with high probability. And second, it constructs $\mathcal{H}$ in expected polynomial time and uses $O(|D|)$ storage. In particular, building on the analysis in \cite{FKS}, the probability for random hash functions $h_r$ and $h_s$ to be successful (i.e., have no collisions) in the first and second stages of \textsc{Exact-Fetch-Prep} can be bounded as follows.

\begin{lemma} (1) For $h_r \in \mathcal{H}_2$,  $Pr[\sum^{16 |D|}_{i=1} b^2_i \leq 4 |D|] \geq \frac{1}{2}$, where $b_i$ corresponds to the counts of collisions in each hash bucket $i$. (2) Functions $h_r,h_s \in \mathcal{H}_2$ succeed with no collisions with a probability $\geq 3/4$. %\label{secondfunction}
\end{lemma}

\begin{proof} Proof of Part (1) is similar to the analysis in ~\cite{FKS}. For Part (2), we have the following.

Define the following random variable:
\begin{align*}
X_i = |&\left\{(x,y)|\; x \neq y, x,y \in \{p_j: 1 \leq j \leq |D|\} \right.,\\
 & \left.h_r(x) = h_r(y) = i, h_s(x) = h_s(y)\right\}|.
\end{align*}

Further, let $X = \sum_i^{16 \cdot |D|} X_i$. If $X > 0$, then the two level hashing of Algorithm \ref{data-structure} fails.  If we show that $Pr [|X| \le 1/2] \geq 1-1/4 = 3/4$, then it would imply that $|X| = 0$ (as $X$ is a natural number) and the two level hashing succeeds with high probability. 

For a randomly chosen hash function $h_s \in \mathcal{H}_2$, we have that for any $x,y$, $Pr [h_s(x) \allowbreak=\allowbreak h_s(y)] = {\frac{1}{16 |D|}}$. We estimate the expected value of random variable $X$:
 
\begin{align*}
E[X_i] \leq {|B_i| \choose 2} {\frac{1}{16 |D|} = \frac{|B_i|(|B_i|-1)}{32 \cdot |D|}},
\end{align*}
where $B_i$ is the hash bucket corresponding to index $i$ with size $b_i$.

Thus, the expected number of colliding pairs summed up over all buckets is  
\begin{align*}
E[X] = E[\sum X_i] \leq {\sum_{i=1}^{16|D|} \frac{ b_i^2}{32 \cdot |D|}}.
\end{align*}
  
Since $\sum_{i=1}^{16|D|} b_i^2 \leq 4|D|$, therefore $E[X] \leq \frac{1}{8}$ and $Pr [X \geq  1/2] \leq \frac{E[X]}{1/2}$. Or, $Pr [X \le 1/2] \geq 1-1/4 = 3/4$.
\end{proof}

\begin{lemma}
The data structure $\mathcal{H}$, can be constructed by \textsc{Exact-Fetch-Prep} in expected polynomial time.
\end{lemma}

\subsection*{Query Stage}

Given $\mathcal{H}$ on the server-side, we do a linear scan on the server-side at query-time to retrieve the consequents of applicable rules. A basic building block that is used in the linear scan is the retrieval of a single element from $\mathcal{H}$. We discuss this first.

To query whether a string $str$ is present in $\mathcal{H}$, one can compute the following quantities: $x = r'\circ str$ ($\circ$ denotes the concatenation operator), and index $i = L\cdot h_r(x) + h_s(x)$. We can then fetch the indexed element $\mathcal{H}[i]$ including one of its attributes $\mathcal{H}[i].C_r(x')$ (here $x'$ corresponds to the string present at location $i$). This attribute can be used to verify if $str$ was indeed present in the database. If a client sends a query for the presence of an element $str$ to a server that only holds $\mathcal{H}$, then the client can have limited privacy. In particular, the server does not know what $str$ is, although it knows the index $i$ and the element was returned (which may not contain the $str$ itself). Since the hashes are not invertible, it affords partial privacy as the client is not revealing its string $str$. \textsc{Exact-Fetch-Query} (Algorithm~\ref{algo:Exact-Fetch-Query}) implements this query process. 

%%%%%%%%%%%%%%%%%%%%%%%%%%%%%%%%%%%%%%%%%%%%%%%%%%%%%%%%%%%%%%%%%%%%%%%%%%%%%%%%%%%%
\begin{algorithm}[h]
\scriptsize
\DontPrintSemicolon
\KwIn{Query string $str$, and data structure $\mathcal{H}$ from Algorithm~\ref{data-structure}}
\KwOut{Value in $\mathcal{H}$ corresponding to query}
\Begin
{
  Client $\client$ computes $x = r' \circ str$, $h_r(x)$ and $ h_s(x)$. \\
  $\client$ computes $i = L \cdot h_s(x) + h_r(x)$.\\
  $\client$ queries server $\server$ for entry at index $i$ in $\mathcal{H}$.
  Server returns $\mathcal{H}$[i].
}

\caption{\textbf{\textsc{Exact-Fetch-Query}($str$)}: Query data structure $\mathcal{H}$}
\label{algo:Exact-Fetch-Query}
\end{algorithm}
\setlength{\textfloatsep}{0pt}
%%%%%%%%%%%%%%%%%%%%%%%%%%%%%%%%%%%%%%%%%%%%%%%%%%%%%%%%%%%%%%%%%%%%%%%%%%%%%%%

The algorithm description for querying the server under the \textsc{TOP-Assoc}($k,w,t,f$) criterion is provided in Algorithm~\ref{algo:Exact-TOP-Assoc} (the algorithms for \textsc{TOP-1-Assoc}, \textsc{TOP-K-Assoc}, \textsc{ALL-Assoc} and \textsc{ANY}-\textsc{Assoc} are similar, hence omitted).

%%%%%%%%%%%%%%%%%%%%%%%%%%%%%%%%%%%%%%%%%%%%%%%%%%%%%%%%%%%%%%%%%%%%%%%%%%%%%%%%%%%%
\begin{algorithm}[h]
\scriptsize
\DontPrintSemicolon
\KwIn{Transaction $\transaction$, data structure $\mathcal{H}$ from Algorithm~\ref{data-structure}, threshold weight $w \in \mathbb{Z}_{+}$, antecedent length parameter $t \in \mathbb{Z}_{+}$, output size parameter $ k \in \mathbb{Z}_{+}$, and ordering function $f$.}
\KwOut{Set of $k$ consequents of applicable association rules $\mathcal{L}$}
\Begin
{
Initialize $\mathcal{L}_{\textrm{all}} \leftarrow \phi$\\
  \textbf{Do parallel}\{\\
  \ForAll{$p_i\subset \transaction$ \& $|p_i|\leq t$ \& $w_i\geq w$}
  {
    $q \leftarrow$ \textsc{Exact-Fetch-Query}$(p_i)$ (Algorithm~\ref{algo:Exact-Fetch-Query}, ignoring the client/server distinction)\\
    \eIf{$q = \phi$}
    {
      continue
    }
    {
      $\mathcal{L}_{\textrm{all}}.\textrm{add}(q)$
    }
  }
  \}\textbf{End parallel}\\
  %return $\mathcal{L}$
  $\mathcal{L} \leftarrow$ first $k$ elements from sort($\mathcal{L}_{\textrm{all}},f$) (break ties arbitrarily if needed)\\
  return $\mathcal{L}$
}
\caption{\textbf{\textsc{Exact-TOP-Assoc-Query}($\transaction,\mathcal{H}$)}: Query with \textsc{TOP-Assoc} criterion}
\label{algo:Exact-TOP-Assoc}
\end{algorithm}
\setlength{\textfloatsep}{0pt}
%%%%%%%%%%%%%%%%%%%%%%%%%%%%%%%%%%%%%%%%%%%%%%%%%%%%%%%%%%%%%%%%%%%%%%%%%%%%%%%%%%%%

\section{Privacy-preserving Protocols}
\label{sec:private-protocols}

We address three related privacy-preserving tasks in sequence. First, we discuss how the \emph{Oblivious transfer} protocol is a solution to privacy-preserving database query problem. We present how consequents of all applicable association rules for a given criterion can be fetched and collated in privacy preserving manner. Finally, we discuss how the protocol can be extended to the setting when we use the approximate algorithms from Section~\ref{sec:approxalgo} for fetching.

\subsection*{Private Protocol for Database Lookup}
  Consider the following two party task: a client $\client$ has an index $i$, and a server $\server$ has a database $D$ represented as a vector $\overrightarrow{v}[1:|D]$ of $|D|$ elements. The client's goal fetch the $i^{th}$ element $\overrightarrow{v}[i]$ such that: (a) the client learns nothing more than the element it fetched from the server, and (b) the server learns nothing about client's query.  Specifically, this leads to the following definition for oblivious transfer (OT):

%What is Sim, ViewC, ViewS?
\begin{definition}
\label{oblivioustransfer}
  An oblivious transfer(OT) protocol is one in which $\client$ retrieves the $i^{th}$ element from $\server$ holding $[1,\dots,n]$ elements \emph{iff} the following conditions hold:
\begin{mylist}
\item The ensembles $View_S(\server(\overrightarrow{v}),\client(i))$ and $View_S(\server(\overrightarrow{v}),\client(j))$ are computationally indistinguishable for all pairs $(i,j)$, where the random variable $View_S$ refers to the transcript of the server created by the execution of the protocol.
\item There is a (probabilistic polynomial time) simulator $Sim$, such that for any query element $c$, the ensembles 
$\allowbreak {Sim(c,\overrightarrow{v}[c])}$ and ${View_C(\server(\overrightarrow{v}),\allowbreak \client(c))}$ are computationally indistinguishable.
\end{mylist}
\end{definition}

  We use the notation OT[$\mathcal{C}:i$, $\mathcal{S}:[1,\ldots,|D|]$] to represent the above Protocol in Definition~\ref{oblivioustransfer}. Without going into the details of OT implementation, we make the design choice to use a fast and parallel implementation described in ~\cite{lipmaaImpl}. This scheme is based on length preserving additive homo-morphic encryption, described next. Homo-morphic encryption with public key $pk$, of message $m$, is denoted as $c=E_{pk}(m)$. Decryption with private key $sk$ is denoted as $m=D_{sk}(c)$. Any operation over the cipher text, will also be reflected in the decrypted plain text. For instance, let $c_1$ and $c_2$ be two cipher texts such that $c_1=E_{pk} (m_1)$ and $c_2=E_{pk} (m_2)$. Let $+$ represent a binary operation. Then, $c_1+c_2=E_{pk} (m_1+m_2)$. Further-more, the scheme is length preserving so that an $l$-bit input is mapped to an input of size $l+N$, where $N$ is a constant.

  We now discuss how to answer the question of whether a string $str$ is in $D$ in a privacy-preserving manner. Recall the data structure $\mathcal{H}$ output by \textsc{Exact-Fetch-Prep} (see Section~\ref{sec:exactalgo}) in which $\client$ fetches an element from database stored with $\server$. We can readily get a stronger privacy-preserving protocol for the same, by employing the OT protocol from Definition~\ref{oblivioustransfer}. That is, by a single execution of OT between client $\client$ and server $\server$, $\client$ can privately fetch the $C_r$-hash of string $r' \circ str$, stored in record
$\mathcal{H}[D \cdot h_r(r' \circ str) + h_s(r' \circ str)].C_r()$. Note that for two strings $str, str'$, it may be that $h_r(r' \circ str) = h_r(r' \circ str')$, and $h_s(r' \circ str) = h_s(r' \circ str')$. Yet the corresponding $C_r$-hashes of $str$ and $str'$ may not be equal. Thus, the choice of using a $C_r$-hash leads to the following guarantee on the OT based protocol.

\begin{lemma}
\label{protocol-guarantee}
  There exists a two party protocol \textsc{Private-Exact-Fetch-Query} such that: (1) $\client$ learns whether $str \in D$ with high probability given the description of associated	hash functions $(h_r, h_s)$, and (2) the computationally bounded $\server$ learns nothing.
\end{lemma}

\subsection*{Private Protocol based on Exact Algorithms}
  A client computes an ordered list of recommended items from a set of consequents of all applicable association rules, chosen according to some selection criteria. We break down the process of making this recommendation process private into the following subtasks.

\begin{mylist}
\item{}	\textit{Expunge infrequent items and anonymize item list:} Firstly, note that only a few items from the client's transaction may be frequent and belong to any rule. So, it is important for the client to remove all infrequent items from its transaction before further processing. The task (denoted \textsc{Preprocess}, see Algorithm~\ref{algo:preprocess}) is to remove  infrequent items, and anonymize the input transaction of the client. We assume that the initial list of items are given identifiers from the range $[|\itemset|]$, which are publicly available (hence available to the client).

\item{} \textit{Privately fetch and privately, interactively collate applicable association rules:} We need to select applicable rules according to the given criterion, and given the consequents of these applicable rules along with their respective weights, we need to privately collate them to produce a list of recommended items using these weights. For this, the client is given a list of identities, with associated weights (which are homo-morphically encrypted), that are obtained from the selection of rules. Client $\client$ and server $\server$ interactively execute a two-party private sorting (denoted \textsc{Private-Two-Party-Sort}, see Algorithm~\ref{algo:P1}) to sort the list of items using their encrypted weights, and the client finally produces an ordered list of recommended items at its end, sorted according to their weights.

\item{} \textit{De-anonymize and recommend:} Given a final list of $k'$ anonymized item identities, we de-anonymize them to obtain the actual names of the recommended items. For this, the client $\client$ fetches their actual identifiers by executing OT (see Definition~\ref{oblivioustransfer}), with the server $\server$ (similar to \textsc{Preprocess} above) on the reverse mapping ($RT$, see Algorithm~\ref{algo:preprocess}), and obtain the true identifiers of the items to be recommended.
\end{mylist}

The above steps are captured in Protocol~\ref{algo:pExact-ALL-Assoc}, which builds on the exact implementation of \textsc{ALL-Assoc}($w,t$) from Section~\ref{sec:exactalgo}. For brevity, we discuss the special case when $t = |\itemset|$. This protocol makes use of the \textsc{Private-Exact-Fetch-Query}($p$) algorithm (see above) as a subroutine. After its execution, the client $\client$ fetches all association rules with weights $\geq w$ in a privacy preserving manner, and from these rules, it collates the list of recommended items. The private versions of the exact implementations of \textsc{TOP-Assoc}($k,w,t$), \textsc{TOP-1-Assoc}($f$), \textsc{TOP-K-Assoc}($f$) and \textsc{ANY-Assoc}($k,w,t$) can be designed in a similar manner. 

We note that the above process ensures privacy of the client data with respect to the server, and the privacy of server's data with respect to the client, by only revealing the relevant consequents of association rules to the client. A much simpler privacy preserving protocol may be devised, if only the privacy of the clients data (transaction $\transaction$) is to be guaranteed, with respect to the server.

%%%%%%%%%%%%%%%%%%%%%%%%%%%%%%%%%%%%%%%%%%%%%%%%%%%%%%%%%%%%%%%%
\begin{protocol}[h]
\scriptsize
\DontPrintSemicolon
\KwIn{$\mathcal{C}$ holds Transaction set $\transaction$ and $\mathcal{S}$ holds itemset $\mathcal{I}$}
\KwOut{Item identifiers}
%\KwOut{$t = |T|$}
\Begin
{
  Server Preprocessing
	\Begin
  {
		$\pi \xleftarrow[permutation]{Random}$ from $[1,\dots,|\itemset|]$.\\
		$T \leftarrow$ table with $|\itemset|+1$ entries\\
		Store $\pi$ in $T$ such that $T[i] = \pi(i)$.\\
		Map item $\mathcal{I}' \in \mathcal{I}$ to $inf$, where $freq(\mathcal{I}')$ \textless $\theta$.\\
    $T[inf] \leftarrow 0$.\\
		Let $RT$ be the reverse map i.e., $RT \circ T (i) = i$, $\forall$ frequent items.\\
	}

	$\mathcal{C}$ has transaction $\transaction=\{i_1,i_2,i_3,\dots,i_{|T|}\}$ comprising of all table entries and $\mathcal{S}$ has the table $T$.\\
	$\mathcal{C}$ and $\mathcal{S}$ execute OT (see Definition~\ref{oblivioustransfer}), with $\mathcal{C}$ seeking table entries $T[i_i, i_2,\dots, i_{|\transaction|}]$\\
	$\mathcal{S}$ randomly permutes the output of the OT before sending to the client $\client$.\\
	$\mathcal{C}$ receives the outputs of the OT from $\mathcal{S}$, decrypts all outputs, and discards all $inf$ entries, corresponding to infrequent items.\\
}
\caption{\bf \textsc{Pre-processing and Anonymization}}
\label{algo:preprocess}
\end{protocol}
\setlength{\textfloatsep}{0pt}
%%%%%%%%%%%%%%%%%%%%%%%%%%%%%%%%%%%%%%%%%%%%%%%%%%%%%%%%%%%%%%%%

%%%%%%%%%%%%%%%%%%%%%%%%%%%%%%%%%%%%%%%%%%%%%%%%%%%%%%%%%%%%%%%%%%%%%%%%%%%%%%%%%%%%
\begin{protocol}[h]
\DontPrintSemicolon
\scriptsize
\KwIn{Client: Transaction $\mathcal{T}$} 
\KwIn{Server: Threshold weight $w$, data structure $\mathcal{H}$ containing $D$ association rules}
\KwOut{Client: Set of recommended items $\mathcal{I}_{\textrm{rec}}$}
\Begin
{
  $\mathcal{C} \leftrightarrow \mathcal{S}$ \textsc{Preprocess} for anonymizing $\mathcal{T}$\\
  $\mathcal{C} \xleftarrow{\mathcal{A}[1..t]} \mathcal{S}$ : Call \textsc{Private-Exact-Fetch-Query} (Lemma~\ref{protocol-guarantee}) for $t=|T|$ times.\\
  $\mathcal{C} \leftrightarrow \mathcal{S}$ : Execute \textsc{Private-Two-Party-Sort} (Algorithm~\ref{algo:P1}) based on whether the weight of the associated rule is $\geq w$.\\
  $\mathcal{C} \xleftarrow{\mathcal{L}} \mathcal{S}$; where $\mathcal{L}$ is the list of association rules sorted according to their weights\\
  $\mathcal{C}$ collates the consequents of rules in $\mathcal{L}$, and calculates $\mathcal{I}_{\textrm{rec}}$\\
}
\caption{\bf \textsc{Private-Exact-ALL-Assoc-Query}($w,t=|\itemset|$)}
\label{algo:pExact-ALL-Assoc}
\end{protocol}
\setlength{\textfloatsep}{0pt}
%%%%%%%%%%%%%%%%%%%%%%%%%%%%%%%%%%%%%%%%%%%%%%%%%%%%%%%%%%%%%%%%%%%%%%%%%%%%%%%%%%%%

\myparagraph{Details of Privacy-preserving Sorting Protocol} We now discuss the details of \textsc{Private-Two-Party-Sort} mentioned above. It is based on a primitive that makes $n^{1.5}$ pairwise comparisons, which are chosen at the pre-processing phase and dependent only on value of $n$ (in our case, $n = |D|$). Using \cite{AKS}, $\mathcal{S}$ produces the identities of the $m=c \cdot n^{1.5}$ pairs, knowing the comparisons of which one can execute the oblivious sorting algorithm \textsc{Private-Two-Party-Sort} in one shot. It is detailed in Protocol~\ref{algo:P1}, and takes only two rounds of communications, with a communication complexity of $O(n^{1.5})$ per round. We recap the properties of the primitive used in the protocol below.

\begin{theorem} [\cite{AKS}] There exists a deterministic pair of algorithms $(AKS_1,\allowbreak AKS_2)$, which satisfy the following:
\begin{mylist}
\item{} Given input $n$, $AKS_1$ produces a list $L_n$ of $O(n^{1.5})$ pair of indexes $(i,j)$.
\item{} Given an input list of integers $I_n = a_1,a_2,\dots,a_i,\dots,a_n$ and value of comparisons of $(a_i,a_j)$ for all $(i,j) \in L_n, |L_n| = n^{1.5}$, deterministic algorithm $AKS_2$ sorts the input list $I_n$.
\end{mylist}
\end{theorem}

%%%%%%%%%%%%%%%%%%%%%%%%%%%%%%%%%%%%%%%%%%%%%%%%%%%%%%%%%%%%%%%%%%%%%%%%%%%%%%%%%%%%
\begin{protocol}[ht]
\scriptsize
\DontPrintSemicolon
\KwIn{$\mathcal{C}$ has pair of values $\mathcal{V}$}
\KwIn{$\mathcal{S}$ has identities}
\KwOut{$\mathcal{C}$ arranges $\mathcal{V}$ in sorted form}
\Begin
{
$\mathcal{C}$\{\\
  Let the pairs of these index of $\mathcal{V}$ be $(x_1,y_1),\dots,(x_m,y_m)$ using $AKS_1$.\\
$Weight(x_i,y_i) \leftarrow (T_{x_i},T_{y_i})$\\
$(T'_{x_i},T'_{y_i}) \xleftarrow{Enc} (T_{x_i},T_{y_i})$\\ 
$(\mathcal{E}_{x_i},\mathcal{E}_{y_i})  \leftarrow (\textsc{Rand}(T'_{x_i}, T'_{y_i})$ (\textsc{Rand} in Algorithm~\ref{algo:randEnc})\\
$P \xleftarrow{sort} (\mathcal{E}_{x_i},\mathcal{E}_{y_i}), i = 1,\ldots,m$\\
\}\\
$\mathcal{C} \xrightarrow{P} \mathcal{S}$\\ 
$\mathcal{S}$: $\mathcal{D} = (d_1,\ldots,d_n), d_i \xleftarrow{Dec}p_i, p_i \in P$ \\
$\mathcal{S}$ : $val \leftarrow$ compare values of $\mathcal{D}$\\ 
$\mathcal{C} \xleftarrow{val} \mathcal{S}$\\ 
$\mathcal{C}$ then applies the $AKS_2$ algorithm to sort $P$.
}
\caption{\textsc{Private-Two-Party-Sort}: Private two party sorting}
\label{algo:P1}
\end{protocol}
\setlength{\textfloatsep}{0pt}
%%%%%%%%%%%%%%%%%%%%%%%%%%%%%%%%%%%%%%%%%%%%%%%%%%%%%%%%%%%%%%%%%%%%%%%%%%%%%%%%%%%%

%%%%%%%%%%%%%%%%%%%%%%%%%%%%%%%%%%%%%%%%%%%%%%%%%%%%%%%%%%%%%%%%%%%%%%%%%%%%%%%%%%%%
\begin{algorithm}[h]
\scriptsize
\DontPrintSemicolon
\KwIn{Data pair ($T_1, T_2$)}
\KwOut{Randomize encrypted data pairs ($\mathcal{E}_1,\mathcal{E}_2$)}
\Begin
{
	Let $T_1, T_2 \in \mathbb{Z}_t$\\
	$(a_1,b_1), (a_2, b_2) \xleftarrow{random} \mathbb{Z}_t$\\
	$\mathcal{E}_1, \mathcal{E}_2\leftarrow (a_1.T_1 + b_1, a_2.T_2 + b_2)$ where $\mathcal{E}\in\mathbb{Z}_t$,
	such that $\mathcal{E}_1, \mathcal{E}_2$ preserves order of $T_1$ and $T_2$ 
}
\caption{\textsc{Rand}: {Randomization of an encrypted data pair}}
\label{algo:randEnc}
\end{algorithm}
\setlength{\textfloatsep}{0pt}
%%%%%%%%%%%%%%%%%%%%%%%%%%%%%%%%%%%%%%%%%%%%%%%%%%%%%%%%%%%%%%%%%%%%%%%%%%%%%%%%%%%%

\subsection*{Private Protocol based on Approximate Algorithms}
\label{subsec:approx}

A few modifications  are needed to the previous protocol when working with data structures designed in Section~\ref{sec:approxalgo}. Recall the functioning of \textsc{Approx-GSCS-Prep} and \textsc{Approx-GSCS-Query}: Server chooses $l$ random maps ($l=\sum_{m=1}^{M}L_m\cdot K_m$), where the $i^{th}$ map $func_i$, maps a set $\transaction$, represented as a characteristic vector $\transaction$ of length $|\itemset|$, to a string $T_i$ of $k$ (say $k=\max_{m=1,...,M}K_m$) bits. Thus, each antecedent $p_i$ of our rules (for $1 \leq i \leq |D|$) is mapped to $l$ strings $p_1,p_2,\dots,p_i,\dots,p_l$ of length $k$ bits each. For an input transaction $\transaction$, an association rule $p \rightarrow q$ is selected if and only if any of the $i$ maps $func_i(\transaction)$ exactly matches $func_i(p)$. We then proceed along the lines of the previous private protocol by doing the following modifications.\\

\noindent \textbf{Pre-processing the $D$ rules}: We create an enhanced database $D_e$ by first choosing $l$ random strings $r_1, r_2, r_3, \dots, r_l \in \{0,1\}^{s}$, where $s$ is a security parameter (a large positive integer). We then concatenate the above random strings to the $l$-maps as follows: $r_1 \circ func_1(p_i), r_2 \circ func_2(p_i), \dots, r_j \circ func_j(p_i),$ $\dots, r_l \circ func_l(p_i)$, for each $1 \leq i \leq |D|$. The new database $D_e$ has $l \cdot |D|$ elements, each of which stores all relevant information for the ARs. All strings $r_i \circ func_1(p_i)$, along with corresponding consequents $q_i$ in $D_e$ are stored in $\mathcal{H}$ as defined by \textsc{Exact-Fetch-Prep} (Algorithm~\ref{data-structure}).\\

\noindent \textbf{Pre-processing the query:} The client $\client$ obtains the definition of the $l$ maps,	$func_i, i \in \{1,\dots,l\}$, along with the random prefixes $r_i$, which are declared publicly. It then applies the $l$ maps on the characteristic vector $\transaction$, corresponding to its input transaction $\transaction$, and computes $func_i(\transaction)$, from which it prepares $r_i \circ func_i(\transaction)$ for $i \in \{1,\dots,l\}$.\\

\noindent \textbf{Privately receiving answers to the query:} $\client$ queries $\server$ for existence of each string $r_i \circ func_i(\transaction)$, for $i = \{1, 2, \dots, l\}$, with the $\server$ possessing the enhanced data structure via \textsc{Private-Exact-Fetch-Query} that is based on OT (Lemma ~\ref{protocol-guarantee}).\\

\myparagraph{Remarks on Security Analysis} The protocols presented here are two party protocols, executed between a client $\client$ and a server $\server$ and assume that both are \emph{honest-but-curious}. That is, they will execute the given protocols, but can be curious to know more (about the inputs of the other party via the transcripts). The security of these protocols can be defined along the standard Ideal/Real paradigm definitions of security, and hinges on the security of the underlying encryption system. In the Ideal/Real paradigm, an ideal functionality is defined which captures the desired input/output functionality for the two parties. First, the parties submit their inputs to a TP (trusted party), and then receive some outputs. The real protocol is said to realize this ideal functionality, if there is a PPT (probabilistic polynomial time) simulator that can compute a distribution of views of the parties that is indistinguishable from the distribution of views generated in the real process. If this is the case, then the protocol is claimed to be secure. In other words, a user can assume that the designed protocol provides security guarantees that one could imagine, as guaranteed by the Ideal functionality.
 
Proofs of security of the presented protocols can be presented along these above lines. For example, in the retrieval of an element from the data structure one would need to capture what is learned by the server about the client's input $\transaction$ and like-wise by the client about server's database of $D$ association rules. This could be the number of consequents fetched, crypto-hash of some antecedent of fetched association rule, size of the client input etc., and this collection is defined as the outputs of the respective parties. A simulation based proof can then proceed along standard lines. Here, we choose to focus on the formulations of the recommendation problem and approaches to solve it efficiently as well as securely and omit the elaborate proofs of security (which are standard because our parties are honest-but-curious, albeit long and relatively less insightful).

Finally, note that in our solutions, for every new client session the server $\server$ has to re-organize the database (use a new set of hash functions in \textsc{Private-Exact-Fetch-Query}). Otherwise, a client may be able to correlate the information received from the server from multiple sessions and gain `unintended to be shared' information about the association rules. This often leads to more computation by the server per client session, and is unavoidable if this type of information leak needs to be prevented. There is always a trade-off between the information shared about the server's database with client and vice-versa, and computations done by the client/server. Our solution chooses one end of this spectrum, and many other choices are equally valid.

\section{Experimental Evaluation}
\label{sec:eval}
 Our empirical evaluations illustrate that although the computational and latency requirements generally increase, privacy properties can still be guaranteed at roughly the same time scale as the non-private counterparts for reasonably sized problem instances. The goal of the experiments is to validate how the recommendation latencies are influenced as a function of the number of rules, the selection criteria and various other problem parameters. For instance, we know that OT does introduce measurable latency between the client and the server. But as we show below, for moderately sized datasets (~10000 rules for instance), the communication overhead is quite manageable (well within a few seconds), which may be appealing for near-real-time advertising (i.e., where time scales of the order of a few seconds are acceptable). 

\myparagraph{Experimental setup and evaluation metrics} Experiments were conducted on a laptop equipped with a $2.5$ GHz Core i5 processor and $16$ gigabyte of memory running Windows $7$.  All algorithms were implemented in \textsc{Java} version $8$ update $60$ with allocated heap space of $8$ gigabyte. We explored the parameters listed in Table~\ref{tab:metrics} over corresponding ranges to evaluate our algorithms and their private versions. RSA modulus size $\mathcal{N}$ is the key size used in the underlying crypto-system. Increasing $\mathcal{N}$ causes significant reduction in performance while increasing security. To assess this trade-off, we ran experiments using both RSA$1024$ and RSA$2048$. All the experiments were executed $1000$ times to compute the amortized execution times.  

\begin{table}
% \scriptsize
	\centering
	%\resizebox{\columnwidth}{!}{%
		\begin{tabular} {|p{4.2cm}|c|c|c|}
		\hline
		\textbf{Parameter} & \textbf{Symbol} & \textbf{Default value} & \textbf{Range} \\\hline
		RSA modulus size & $\mathcal{N}$ & $1024$ & $\{1024, 2048\}$ \\\hline
		Number of rules & $|D|$ & $10^3$ & $10^3-10^6$ \\\hline
		\textsc{TOP-Assoc} output size & $k$ & $3$ & $3-20$ \\\hline
		\textsc{TOP-Assoc} length & $t$ & $3$ & $3-10$ \\\hline
		Query size & $|\mathcal{T}|$ & $5$ & $5-20$ \\\hline
		\end{tabular}%
	%	}
	\caption{Evaluation Metrics}
	\label{tab:metrics}
\end{table}
\setlength{\textfloatsep}{0pt}

\subsection*{Evaluating the Exact and Approximate Algorithms} 

First, an exact implementation of the \textsc{ALL-Assoc}($w,t$) criterion is evaluated. This criterion was picked because the size of the list of rules output by the exact algorithm is larger than the outputs of the other criteria. Further, the computational burden imposed by parameter $k$ for any $k < |D|$ is negligible in terms of the total processing time. We generated synthetic datasets and evaluated our implementation, whose median processing times are listed in Table~\ref{tab:time}. 
As can be inferred, even when the number of association rules is very large (for instance, see the entry corresponding to $|D|$ equal to $1$ million), our implementation is observed to be very efficient.

\begin{table}
	\centering
	\scriptsize
	\resizebox{\columnwidth}{!}{%
		\begin{tabular} {|c|c|c|c||c|c|c|c||c|c|c|c|}
		\hline
		
  		$|D|$ & $|\mathcal{T}|$ & $t$ & Time &
  		$|D|$ & $|\mathcal{T}|$ & $t$ & Time &  
  		$|D|$ & $|\mathcal{T}|$ & $t$ & Time\\\hline

		\multirow{8}{*}{$100$K} & \multirow{2}{*}{$5$}  & $3$ & $58$  & 
		\multirow{8}{*}{$500$K} & \multirow{2}{*}{$5$}  & $3$ & $229$ &
		\multirow{8}{*}{$1$M} 	& \multirow{2}{*}{$5$} 	& $3$ & $744$ \\
		\cline{3-4}\cline{7-8}\cline{11-12}
		& & $5$ & $69$ & & & $5$ & $273$ & & & $5$ & $931$\\
		\cline{2-4}\cline{6-8}\cline{10-12}

		& \multirow{2}{*}{$10$}	& $3$ & $278$ & & \multirow{2}{*}{$10$}	& $3$ & $436$ & & \multirow{2}{*}{$10$}	& $3$ & $1279$\\
		\cline{3-4}\cline{7-8}\cline{11-12}

		& & $5$ & $232$ && & $5$ & $525$ & & & $5$ & $1413$\\
		\cline{2-4}\cline{6-8}\cline{10-12}

		& \multirow{2}{*}{$15$}	& $3$ & $212$ && \multirow{2}{*}{$15$}	& $3$ & $678$ & & \multirow{2}{*}{$15$}	& $3$ & $1911$\\
		\cline{3-4}\cline{7-8}\cline{11-12}

		& & $5$ & $213$ && & $5$ & $751$ & & & $5$&  $2009$\\
		\cline{2-4}\cline{6-8}\cline{10-12}

		& \multirow{2}{*}{$20$}	& $3$ & $263$ && \multirow{2}{*}{$20$}	& $3$ & $999$ & & \multirow{2}{*}{$20$}	& $3$ & $2615$\\
		\cline{3-4}\cline{7-8}\cline{11-12}

		& & $5$ & $229$ && & $5$ & $987$ & & & $5$ & $2957$\\
		\hline
		\end{tabular}%
		}
	\caption{Execution times (in milliseconds) incurred by the \textsc{Exact-ALL-Assoc}($w,t$) algorithm for fetching applicable rules. Symbols $K$ and $M$ denotes values $10^3$ and $10^6$ respectively.}
	\label{tab:time}
\end{table}

Second, we evaluated \textsc{Approx-GSCS-Query}($\transaction,\mathcal{DS},f$) on two real world transaction datasets: (a) Retail~\cite{retail}, and (b) Accidents~\cite{accidents}. The retail dataset consists of market basket data collected from an anonymous Belgian retail store for approximately 5 months during the period 1999-2000. The number of transactions is $88163$ and the number of items is $16470$. We use SPMF's~\cite{SPMF} implementation of FP Growth algorithm (setting minimum support and minimum confidence values to $0.001$ and $0.01$ respectively) to get $16147$ association rules. The Accidents dataset consists of traffic accidents during the period 1991-2000 in Flanders, Belgium. The number of transactions in this dataset is $340,184$. The attributes capture the circumstances in which the accidents have occurred. The total number of attributes (corresponding to $\itemset$) is $572$. Again, we use SPMF's~\cite{SPMF} implementation of FP Growth algorithm (setting minimum support and minimum confidence values to $0.5$ and $0.6$ respectively) to get $334566$ association rules.  

Note that because the quality of the recommendations depends on the quality of the rules mined, the recommendation accuracy can vary across datasets. For the above two datasets that we used, we held out some of the transactions in a validation set and checked if the recommendations given using the \textsc{ALL-Assoc} criterion had acceptable validation set precision (fraction of items that were correct among the recommended items). We used the confidence of the rules as their weights. We did not cross-validate to get the best parameter choices (weight threshold $w$ or the antecedent length threshold $t$ for our criterion here) as this was not our primary goal.

Instead, we show how the latency overhead due to the private protocol varies as a function of our system parameters below. In particular, column $T_{q}$ of table \ref{tab:approx-and-priv} lists the median processing times for a collection of predefined query transactions for these two datasets (the thresholding parameter for antecedent length, $t$, was varied between $1$ to $5$) in the non-privacy setting. Contrast this with the column $T_{o}$, which lists the median processing times with privacy on a sub-sampled set of rules as shown in column $|D_o|$. The reduced size of the set of rules considered for private fetching is needed to ensure that the latencies are manageable. The sub-sampling was based on the confidence weights of the ARs (rules with higher weights were picked). As can be inferred, these processing times are comparable to the numbers in Table \ref{tab:time} in an absolute sense (seconds vs milliseconds). Columns $A_{10}$, $A_{16}$ and $A_{32}$ provide accuracies of \textsc{Approx-GSCS-Query} with hash lengths ($k = \max_{m=1,...,M}K_m$, see Section~\ref{sec:approxalgo}) set to $10$, $16$ and $32$ bits respectively, averaged over $1000$ queries of length $3$. When the length is only $10$ bits, \textsc{Approx-GSCS-Query} suffers from low accuracy as irrelevant association rules fall under the same buckets. We omit more extensive results on the approximation quality for brevity (see~\cite{shrivastava2015asymmetric,neyshabur2015symmetric} for extensive performance profiling of similar approximation schemes).

\begin{table}
\centering
% \small
% \resizebox{\columnwidth}{!}{%
		\begin{tabular} {|c|c|c|c|c|c|c|c|c|}
		\hline
		Dataset & $|D|$ & $T_{q}$ (ms) & $|D_o|$&$T_{o}$ (s)& $A_{10}$ & $A_{16}$ & $A_{32}$\\\hline
		Accidents & $334$K & $1.69$ & $10$K &$27.4$& $68\%$ & $100\%$ & $100\%$\\\hline
		Retail & $16$K & $1.5$ & $10$K &$28$& $72\%$& $100\%$ & $100\%$\\\hline
\end{tabular}%
% }
\caption{Performance of \textsc{Approx-GSCS-Query}. \label{tab:approx-and-priv}}
\end{table}

\subsection*{Evaluating Timing Overheads due to Privacy}

Table~\ref{tab:priv} documents the timing overhead introduced by a single $1$-$n$ oblivious transfer, which is used to make the exact implementation for \textsc{ALL-Assoc} privacy preserving. The number of rules ($|D|$) was varied from $1$ thousand to $10$ thousand and the RSA modulus was varied between $1024$ and $2048$. Our implementation of the oblivious transfer protocol is single threaded and is based on~\cite{LipmaaLogSq} (a multi-threaded faster implementation can be found in ~\cite{lipmaaImpl}, which can be used as a plug-in module to improve overhead time by a factor of magnitude or higher). From the table we can infer that for moderate sized databases, private fetching of applicable rules is competitive and practical. For instance, to fetch applicable rules from a database with $10^4$ rules, the median time taken is $\sim 40$ seconds for a query of size $5$ (RSA modulus set to $1024$). Practicality is further supported by the fact that in client server settings applicable for many cloud based applications/world wide web, multiple servers will be handling multiple query requests.

\begin{table}
	\centering
	\scriptsize
	\resizebox{\columnwidth}{!}{%
		\begin{tabular} {|c|c|c|c||c|c|c||c|c|c|c||c|c|c|}
		\hline
		
		$\mathcal{N}$ & $|D|$ & $|\transaction|$ & Time & $|D|$ & $|\transaction|$ & Time & $\mathcal{N}$ & $|D|$ & $|\transaction|$ & Time & $|D|$ & $|\transaction|$ & Time\\\hline
		
		\multirow{8}{*}{$1024$} 	& \multirow{4}{*}{$1$K} & $5$  & $3.5$ &  \multirow{4}{*}{$2$K} & $5$  & $6.9$ & \multirow{8}{*}{$2048$} 	& \multirow{4}{*}{$1$K} & $5$  & $25.3$ &  \multirow{4}{*}{$2$K} 	& $5$  & $52.1$\\\cline{3-4}\cline{6-7}\cline{10-11}\cline{13-14}
															& 											& $10$ & $3.6$ & 												& $10$ & $7$  && 											& $10$ & $25.5$ & 												& $10$ & $61$\\\cline{3-4}\cline{6-7}\cline{10-11}\cline{13-14}
															& 											& $15$ & $3.7$ & 												& $15$ & $7.2$&& 											& $15$ & $26$ & 												& $15$ & $50$\\\cline{3-4}\cline{6-7}\cline{10-11}\cline{13-14}
															& 											& $20$ & $3.5$ &                    		& $20$ & $6.8$&& 											& $20$ & $25.8$ &                    			& $20$ & $50.4$\\\cline{2-7}\cline{9-14}
															
															& \multirow{4}{*}{$5$K} & $5$  & $17$  	& \multirow{4}{*}{$10$K}& $5$  & $39.6$&& \multirow{4}{*}{$5$K} & $5$ & $133.1$ & \multirow{4}{*}{$10$K}& $5$  & $266.9$\\\cline{3-4}\cline{6-7}\cline{10-11}\cline{13-14}
															& 											& $10$ & $18.5$ &												& $10$ & $37.6$&& 											& $10$ & $138.7$ &												& $10$ & $248.6$\\\cline{3-4}\cline{6-7}\cline{10-11}\cline{13-14}
															& 											& $15$ & $17.5$ &												& $15$ & $39.1$&& 											& $15$ & $129.8$ &												& $15$ & $256.5$\\\cline{3-4}\cline{6-7}\cline{10-11}\cline{13-14}
															& 											& $20$ & $17.2$ &												& $20$ & $36.5$&& 											& $20$ & $141.8$ &												& $20$ & $268.5$\\\hline 
	\end{tabular}%
	}

	\caption{Overhead times by a private protocol (see Section~\ref{sec:private-protocols}) that embeds the exact implementation of \textsc{ALL-Assoc}($w,t$).}
	\label{tab:priv}
\end{table}

As discussed briefly earlier, the timing overheads incurred by the privacy preserving counterpart of \textsc{Approx-GSCS-Query} for the real datasets is shown in Table \ref{tab:approx-and-priv} (column $T_{o}$). We choose the RSA modulus value to be $1024$ here. Although the processing times are now multiple orders of magnitude compared to vanilla processing times (column $T_{q}$), they are still practical and manageable for an e-commerce setting (again, due to the fact that in practice multiple servers service multiple queries). These times are also comparable to similar sized datasets benchmarked in Table \ref{tab:priv}. Thus, our solutions and their private versions are very competitive in fetching applicable association rules and making item recommendations.  %Note that we have excluded the processing times incurred for the Collate step because the processing times required for getting recommended items from selected applicable rules were much lower than those for the Fetch step.
\section{Conclusion}
\label{sec:conc}

Our work proposes a rich set of methods for selection and application of association rules for recommendations that have strong theoretical basis as well as pragmatic grounding. The ability to reuse association rules that are frequently used in the industry to bootstrap \e{a} scalable and privacy-aware recommendation system makes our solution very attractive to practitioners. Our experiments further highlight the practicality of achieving privacy preserving recommendations for moderate to large-scale e-commerce applications.

\bibliographystyle{myIEEEtran}
\bibliography{arxiv_privacy}
%\appendix
%\input{appendix_nonprivate}
%\input{appendix_private}

%\input{appendix-A}
\end{document}